\documentclass[12pt]{article}
\RequirePackage[colorlinks,citecolor=blue,urlcolor=blue,linkcolor=blue]{hyperref}
\usepackage{url}

\usepackage[ansinew]{inputenc}

\usepackage{doi}

\usepackage{amsmath}

\usepackage{mathtools}

\usepackage{amsfonts}
\usepackage{amsthm}
\usepackage{amssymb}

\usepackage{color}
\usepackage{bm} 
\usepackage{dsfont} 

\usepackage{graphicx}
\usepackage[lmargin=0.5in,rmargin=0.5in,bottom=1in,top=1in]{geometry}

\usepackage{tikz}
\definecolor{mycolor1}{rgb}{0.105882,0.619608,0.466667}
\definecolor{mycolor2}{rgb}{0.85098,0.372549,0.00784314}
\definecolor{mycolor3}{rgb}{0.458824,0.439216,0.701961}
\definecolor{mycolor4}{rgb}{0.905882,0.160784,0.541176}
\definecolor{mycolor5}{rgb}{0.4,0.65098,0.117647}
\definecolor{mycolor6}{rgb}{0.65098,0.462745,0.113725}
\definecolor{mycolor7}{rgb}{0.901961,0.670588,0.00784314}
\definecolor{mycolor8}{rgb}{0.4,0.4,0.4}
\definecolor{mycolor9}{rgb}{0.301961,0,0.294118}
\definecolor{mycolor10}{rgb}{0.0313725,0.25098,0.505882}

\usetikzlibrary{calc}        

\makeatletter
\newif\ifmygrid@coordinates
\tikzset{/mygrid/step line/.style={line width=0.80pt,draw=gray!80},
         /mygrid/steplet line/.style={line width=0.25pt,draw=gray!80}}
\pgfkeys{/mygrid/.cd,
         step/.store in=\mygrid@step,
         steplet/.store in=\mygrid@steplet,
         coordinates/.is if=mygrid@coordinates}
\def\mygrid@def@coordinates(#1,#2)(#3,#4){%
    \def\mygrid@xlo{#1}%
    \def\mygrid@xhi{#3}%
    \def\mygrid@ylo{#2}%
    \def\mygrid@yhi{#4}%
}
\newcommand\DrawGrid[3][]{%
    \pgfkeys{/mygrid/.cd,coordinates=true,step=1,steplet=0.2,#1}%
    \draw[/mygrid/steplet line] #2 grid[step=\mygrid@steplet] #3;
    \draw[/mygrid/step line] #2 grid[step=\mygrid@step] #3;
    \mygrid@def@coordinates#2#3%
    \ifmygrid@coordinates%
        \draw[/mygrid/step line]
        \foreach \xpos in {\mygrid@xlo,...,\mygrid@xhi} {%
          (\xpos,\mygrid@ylo) -- ++(0,-3pt)
                              node[anchor=north] {$\xpos$}
        }
        \foreach \ypos in {\mygrid@ylo,...,\mygrid@yhi} {%
          (\mygrid@xlo,\ypos) -- ++(-3pt,0)
                              node[anchor=east] {$\ypos$}
        };
    \fi%
}
\makeatother

\usepackage{psfrag}

\usepackage{algorithm}
\usepackage{algpseudocode}
\algdef{SE}[DOWHILE]{Do}{doWhile}{\algorithmicdo}[1]{\algorithmicwhile\ #1}%

\usepackage{mathrsfs} 

\usepackage[square,numbers,compress]{natbib}

\newcommand{\remove}[1]{}
\newcommand{\removesafe}[1]{}

\usepackage{subfigure}

\newcommand{\R}{\mathbb{R}}
\newcommand{\A}{\mathcal{A}}
\newcommand{\eps}{\epsilon}
\newcommand{\xtilde}{\tilde{x}}
\newcommand{\PP}{\mathbb{P}}

\newtheorem{theorem}{Theorem} 
\newtheorem{lemma}[theorem]{Lemma}

\title{An Elementary Proof of Convex Phase Retrieval in the Natural Parameter Space via the Linear Program PhaseMax}
\author{
Paul Hand\footnote{Department of Computational and Applied Mathematics, Rice University, TX.}  \ and Vladislav Voroninski\footnote{Helm.ai, CA}}

\begin{document}

\maketitle
\abstract{The phase retrieval problem has garnered significant attention since the development of the PhaseLift algorithm, which is a convex program that operates in a lifted space of matrices.  Because of the substantial computational cost due to lifting, many approaches to phase retrieval have been developed, including non-convex optimization algorithms which operate in the natural parameter space, such as Wirtinger Flow. Very recently, a convex formulation called PhaseMax has been discovered, and it has been proven to achieve phase retrieval via linear programming in the natural parameter space under optimal sample complexity.  The current proofs of PhaseMax rely on statistical learning theory or geometric probability theory.  Here, we present a short and elementary proof that PhaseMax exactly recovers real-valued vectors from random measurements under optimal sample complexity.  Our proof only relies on standard probabilistic concentration and covering arguments, yielding a simpler and more direct proof than those that require statistical learning theory, geometric probability or the highly technical arguments for Wirtinger Flow-like approaches.}

\section{Introduction}
Consider the classical phase retrieval problem: recover $x_0 \in \mathbb{C}^n$ modulo global phase from a collection of measurements $|\langle a_i, x_0 \rangle|, i=1, \ldots m$, where $a_i \in \mathbb{C}^n$ are known measurement vectors. 
We consider the idealized model $a_i \sim \mathcal{N}(0, I_{n \times n})$ for the measurement vectors. The first algorithm which guaranteed exact and noise-robust phase retrieval from $O(n)$ random phaseless measurements was PhaseLift \cite{CESV2011, CSV2013, CL2012}, which is a convex semidefinite program that operates in a lifted space of $n\times n$ matrices.  Because lifting squares the problem dimensionality, PhaseLift is prohibitively computationally expensive and thus much research has focused on alternative formulations. In particular, non-convex formulations include the Wirtinger Flow \cite{wirtinger}, Truncated Wirtinger Flow \cite{twf}, Truncated Amplitude Flow \cite{yoninaTAF}, and the gradient descent of \cite{SQW2016}.  These methods permit phase retrieval from Gaussian measurements under optimal or near optimal sample complexity in the natural parameter space, however they are technically challenging to analyze theoretically and lead to complicated algorithmic formulations with many parameters.  

To the surprise of the community, a recent successful formulation for phase retrieval called PhaseMax, independently developed in \cite{phasemax, phasemaxJustin}, is convex and operates in the natural $n$-dimensional parameter space.  The theoretical results of \cite{phasemax} achieve a tighter sample complexity than those in \cite{phasemaxJustin}, notably providing guarantees very close to the information theoretic lower-bounds.  Both of these approaches rely on first finding an anchor vector which is positively correlated with the vector $x_0$, provided for instance by using a spectral initialization first reported on by Netrapali et al. \cite{altminphase} and further enhanced by authors of Wirtinger Flow-like methods. The proof in \cite{phasemaxJustin} uses arguments based on statistical learning theory, and the the proof in \cite{phasemax} uses arguments from sphere covering and geometric probability.

In this short paper, we consider only the real-valued case for simplicity (the complex case is very similar) and present an alternate elementary proof that  PhaseMax succeeds in finding $x_0$ up to global sign from $O(n)$ phaseless Gaussian measurements, thus achieving phase retrieval under optimal sample complexity via a linear program with a linear number of constraints.  Our proof is based on standard elementary probabilistic concentration estimates of the singular values of random matrices.

\section{Main Result and Proof}
Our main result is that PhaseMax  succeeds at recovering a fixed signal $x_0\in \R^n$ with high probability from $O(n)$ random measurements, when provided with an anchor vector $\phi$ that is sufficiently close to $x_0$. 

\begin{theorem} \label{main-theorem}
Fix $x_0 \in \R^n$.  Let $a_i$ be i.i.d. $\mathcal{N}(0, I_{n \times n})$, for $i = 1\ldots m$.  Let $y_i = |\langle a_i, x_0\rangle |$.   
Let $\phi \in \mathbb{R}^n$ be such that $\|\phi-x_0\|_2 < 0.6 \|x_0\|_2$.  If $m \geq c n$, then with probability at least $1 - 6 e^{-\gamma m}$, $x_0$ is the unique solution of the linear program PhaseMax:
\begin{equation}\label{PhaseMax}
\begin{array}{ll}
\max & \langle \phi, x \rangle \\
\text{s.t.} & -y_i \leq \langle a_i, x \rangle \leq y_i, \quad  i = 1\ldots  m\\ 
\end{array}
\end{equation}
Here, $\gamma$ and $c$ are universal constants.
\end{theorem}

A satisfactory anchor vector can be efficiently computed with high probability by several methods.  For concreteness, consider the truncated spectral initializer in \cite{twf}.  With $\lambda = \sqrt{\frac{1}{m} \sum_{i=1}^m y_i^2}$,
 let 
 \begin{align}
 \phi = \sqrt{\frac{mn}{\sum_{i=1}^m \|a_i\|_2^2}} \lambda z, \text{ where $z$ is the leading eigenvector of }
\frac{1}{m}\sum_{i=1}^m y_i^2 a_i a_i^\intercal 1_{\{|y_i| < 3 \lambda \}}. \label{eq:initializer}
\end{align}
By Proposition 3 from \cite{twf}, for a fixed $x_0$, this truncated spectral initializer satisfies $\min (\|\phi - x_0\|, \|\phi + x_0\|) \leq 0.6 \|x_0\|_2$ with probability at least $1 - e^{-\gamma m}$, provided that $m \geq c_0 n$.   Note that in the case where $\|\phi + x_0\| \leq 0.6 \|x_0\|_2$, then the output of PhaseMax will be $-x_0$ with high probability, which is exact up to the inherent global phase ambiguity.  Alternatively, one could use the leading eigenvector of $\frac{1}{m} \sum_{i=1}^m y_i^2 a_i a_i^\intercal$ as the anchor vector, as done in the initialization steps of AltMinPhase \cite{altminphase}  and Wirtinger Flow \cite{wirtinger}.  In this case $m = O (n \log n)$ measurements are necessary to obtain an accurate anchor vector \cite{twf}.    The initialization from the Truncated Amplitude Flow \cite{yoninaTAF} could also be used under $m = O(n)$.

The entire proof of Theorem \ref{main-theorem} (with the stronger assumption $\|\phi - x_0\| < 0.4 \|x_0\|_2$) can fit on half a page by appealing to standard concentration estimates on the singular values of Gaussian matrices and to an $\ell_1$-isometry bound proved in Lemma 3.2 from PhaseLift \cite{CSV2013}.  For the sake of exposition, we include a direct proof of this isometry bound with a superior constant.

\section{Proof}
Throughout the proof, the values of constants $c$ and $\gamma$ may change line to line, but they are bounded from above and below by fixed positive numbers.  The proof uses a technical lemma that bounds the singular values of $\frac{1}{m} \sum_{i=1}^m a_i a_i^\intercal$ and a technical lemma that bounds $\frac{1}{m} \sum_{i=1}^m |\langle a_i, x_0\rangle| |\langle a_i, x_1\rangle|$ from below with high probability.

\begin{proof}[Proof of Theorem \ref{main-theorem}]
Without loss of generality, take $\|x_0\|_2 = 1$.  To prove $x_0$ is the unique maximizer of \eqref{PhaseMax}, it suffices to show that for any nonzero $h \in \R^n$,
\[
\langle a_i, h \rangle \langle a_i, x_0 \rangle \leq 0 \quad \text{for} \quad i = 1,2\ldots m \quad \implies \quad \langle \phi, h \rangle < 0
\]
Suppose $\langle a_i, h \rangle \langle a_i, x_0 \rangle \leq 0$ for all $i$.  Then,
\[
-\frac{1}{m}\sum_{i=1}^m |\langle a_i, h \rangle \langle a_i, x_0 \rangle | = \frac{1}{m}\sum_{i=1}^m \langle a_i, h \rangle \langle a_i, x_0 \rangle  = \left \langle \frac{1}{m} \sum_{i=1}^m a_i a_i^\intercal , hx_0^\intercal \right \rangle
\]
By Lemma \ref{lem:isometry}, if $m \geq \eps^{-2} n$, then on an event of probability at least $1-2e^{-\gamma \eps^2 m}$,
\[
\left \langle \frac{1}{m} \sum_{i=1}^m a_i a_i^\intercal , hx_0^\intercal \right \rangle \geq \langle h, x_0 \rangle - \epsilon\|h\|_2\|x_0\|_2 \text{ for all $h\in \R^n$}.
\]
By  Lemma \ref{lemma:lower-bound}\footnote{The same result holds with the constant $0.45$  by applying Lemma 3.2 from PhaseLift \cite{CSV2013} to $h x_0^\intercal + x_0 h^\intercal$.}, if $m \geq c_0 n$, then on an event of probability at least $1 - 4 e^{-\gamma m}$,
\[
-\frac{1}{m}\sum_{i=1}^m |\langle a_i, h \rangle \langle a_i, x_0 \rangle | \leq -0.63 \|h\|_2 \|x_0\|_2 \text{ for all $h\in \R^n$}.
\]
Thus, taking $\eps = 0.03$,
\[
\langle h, x_0 \rangle \leq (-0.63+\eps) \|h\|_2 \|x_0\|_2 = -0.6 \|h\|_2\|x_0\|_2.
\]
Finally, for nonzero $h$, we use the assumption that $\| \phi - x_0 \|_2 < 0.6 \|x_0\|_2$ to conclude
\[
\langle \phi, h\rangle = \langle \phi - x_0, h\rangle + \langle x_0, h \rangle \leq \| \phi - x_0\|_2 \|h\|_2 - 0.6 \|h\|_2 <  0
\]
on an event of probability at least $1 - 6 e^{-\gamma m}$. 
\end{proof}

We now prove the technical lemmas.  
\begin{lemma}\label{lem:isometry}
Fix $\eps < 1$.  Let $a_i \in \mathbb{R}^n, i =1,2\ldots m$ be i.i.d gaussian.  There exists a universal constant $\gamma>0$ such that $\left\| \frac{1}{m} \sum_{i=1}^m a_i a_i^* - I_{n\times n} \right\| \leq \epsilon $ with probability at least $1-2 e^{-\gamma \eps^2 m}$, provided $m \geq \eps^{-2} n$.
\end{lemma}
\begin{proof}
This claim follows from standard concentration estimates for Gaussian matrices (e.g. Corollary 5.35 in \cite{V2012}).
\end{proof}


\begin{lemma}  \label{lemma:lower-bound}
There exist constants $c_0, \gamma, \delta_0$ such that for any $0<\delta < \delta_0$, if $m \geq c_0 (\delta^{-2} \log \delta^{-1}) n$, then with probability at least $1-4e^{-\gamma m \delta^2}$
\begin{align}
\frac{1}{m} \sum_{i=1}^m | \langle a_i , x_0 \rangle \langle a_i,  x_1 \rangle| \geq \frac{2}{\pi} (1-\delta) \| x_0\|_2 \|x_1\|_2 \quad \text{ for all } x_0, x_1 \in \R^n \label{eq:lower-bound}
\end{align}
\end{lemma}
\begin{proof} 
Without loss of generality, we take $\|x_0\| = \|x_1\| = 1$. 
Define $\A: \R^{n \times n} \to \R^m, X \mapsto (a_i^\intercal X a_i)_{i=1}^m$.  Observe
 that $\frac{1}{m} \sum_{i=1}^m | \langle a_i, x_0\rangle \langle a_i, x_1\rangle| = \frac{1}{m}\| \mathcal{A}(x_0 x_1^\intercal) \|_1$.  
 
First, we show that for any fixed $x_0, x_1$ of unit length, 
$\frac{1}{m} \| \A(x_0 x_1^\intercal)\|_1 \geq \frac{2}{\pi} - \delta_1$
with probability at least $1 - 2 e^{-\gamma m \delta_1^2}$.  To show this, observe that $\frac{1}{m}\|\A(x_0 x_1^\intercal)\|_1 = \frac{1}{m} \sum_{i=1}^m \xi_i$, where $\xi_i =| \langle a_i, x_0\rangle \langle a_i, x_1\rangle| $ is a subexponential random variable.  Let $K$ be the subexponential norm of $\xi_i$.  By Lemma~\ref{lemma:expectation}, $\mathbb{E} \xi_i \geq \frac{2}{\pi}$.     Thus, by the Bernstein-type inequality (Corollary 5.17 in \cite{V2012}) 
$$
\PP\Bigl( \frac{1}{m} \| \A(x_0 x_1^\intercal)\|_1 \geq \frac{2}{\pi} - \delta_1 \Bigr) \geq 1- 2 \exp \left[- \gamma m \min \Bigl(\frac{\delta_1^2}{K^2}, \frac{\delta_1}{K} \Bigr) \right]
$$
Later, we will take $\delta_1 < K$. 

Second, we show that the bound \eqref{eq:lower-bound} holds simultaneously for all $x_0, x_1$.  Let $\mathcal{N}_\eps$ be an $\eps$-net of $S^{n-1}$.  By Lemma 5.2 in \cite{V2012}, we may take $| \mathcal{N}_\eps| \leq (1 + \frac{2}{\eps})^n$.   For any $(x_0, x_1)$, there exists a $(\xtilde_0, \xtilde_1) \in \mathcal{N}_\eps \times \mathcal{N}_\eps$ such that 
$
\|x_0 x_1^\intercal - \xtilde_0 \xtilde_1^\intercal\|_1  
\leq 3 \eps.
$
Because $\A(x_0 x_1^\intercal) = \A(\xtilde_0 \xtilde_1^\intercal) + \A(x_0 x_1^\intercal - \xtilde_0 \xtilde_1^\intercal)$, we have
$$
\frac{1}{m}\|\A(x_0 x_1^\intercal)\|_1 \geq \frac{1}{m} \| \A(\xtilde_0 \xtilde_1^\intercal) \|_1 - \frac{1}{m}\|\A(x_0 x_1^\intercal - \xtilde_0 \xtilde_1^\intercal)\|_1.
$$
On the intersection of the events 
\begin{align*}
E_1 &= \{m^{-1} \|\A(x_0 x_1^\intercal) \|_1 \geq \frac{2}{\pi} - \delta_1, \text{ for all } (x_0, x_1) \in \mathcal{N}_\eps \times \mathcal{N}_\eps\}\\
E_2 &= \{m^{-1}\| \A(X) \|_1 \leq (1 + \delta_2) \|X\|_1, \text{ for all } X \in \R^{n \times n} \}
\end{align*}
we have
\begin{align*}
\frac{1}{m}\|\A(x_0 x_1^\intercal)\|_1  \geq \Bigl(\frac{2}{\pi} - \delta_1 \Bigr) - (1 + \delta_2)  3 \eps
\end{align*}
Choose $\delta_1 = \delta/\pi, \delta_2 = 1/2, \eps = 2\delta/(9\pi)$.  Further,  choose $\delta_0$ such that $\delta_1 < K$.    Thus,  on $E_1 \cap E_2$, $\frac{1}{m}\|\A(x_0 x_1^\intercal)\|_1  \geq \frac{2}{\pi}(1-\delta)$. It remains to estimate the probability of $E_1\cap E_2$.   We have $$\PP(E_1) \geq 1 - | \mathcal{N}_\eps \times \mathcal{N}_\eps| \cdot 2 e^{-\gamma m \delta_1^2} \geq 1 - 2 \Bigl( 1 + \frac{2}{\eps}\Bigr)^{2n} e^{- \gamma m \delta_1^2}.$$  If $m \geq c_0 n \delta
_1^{-2} \log \delta_1^{-1}$ for a sufficiently large $c_0$,  then $\PP(E_1) \geq 1 - 2 e ^{-\gamma m \delta_1^2}$.  Using standard concentration estimates on the singular values of  Gaussian matrices (Corollary 5.35 in \cite{V2012}), it can be shown that $m \geq 16 \delta_2^{-2} n$ implies $\PP(E_2) \geq 1 - 2 e^{-\gamma m \delta_2^2}$.   See Lemma 3.1 in \cite{CSV2013} for the elementary details.  Thus, $\PP(E_1 \cap E_2) \geq 1 - 4 e^{-\gamma m \delta^2}$ for some constant $\gamma$. 

\end{proof}

\begin{lemma}  \label{lemma:expectation}
Fix $x_0, x_1 \in \mathbb{R}^n$ and let $\theta$ be the angle between $x_0$ and $x_1$.  Let $a \sim \mathcal{N}(0, I_{n \times n})$.  $$\mathbb{E}|\langle a_i, x_0\rangle \langle a_i, x_1 \rangle | = \frac{2}{\pi} \left(|\sin \theta| + \sin^{-1} (\cos \theta) \cos \theta \right) \|x_0\|_2 \|x_1\|_2 \geq \frac{2}{\pi} \|x_0\| \|x_1\|$$
\end{lemma}
\begin{proof}
Without loss of generality, take $\|x_0\|=\|x_1\| = 1$.  Further, without loss of generality, take $x_0 = e_1$ and $x_1 = \cos \theta \ e_1 + \sin \theta \ e_2$.  The expected value is 
\begin{align*}
\mathbb{E} |a_1 (a_1 \cos \theta + a_2 \sin \theta) | &= \frac{1}{2 \pi} \int_0^\infty r^3 e^{-r^2/2} dr \int_0^{2\pi} |\cos \phi \cos(\theta - \phi)| d\phi \\
&=\frac{1}{2\pi} \int_0^{2\pi}|\cos \theta + \cos(2 \phi - \theta)| d\phi\\
&=\frac{1}{\pi} \int_0^\pi |\cos \theta + \cos \tilde{\phi}| d\tilde{\phi}\\
&=\frac{2}{\pi}  \left(|\sin \theta| + \sin^{-1} (\cos \theta) \cos \theta \right) \geq \frac{2}{\pi},
\end{align*}
where the second equality is because $2 \cos \phi \cos(\theta - \phi) = \cos \theta + \cos(2 \phi - \theta)$. 
\end{proof}

\section{Acknowledgements}
PH acknowledges funding by the grant NSF DMS-1464525. 
\bibliographystyle{plain}
\bibliography{refs}

\end{document}